\documentclass{theoretics}

\usepackage{array} 

\usepackage{amsfonts, verbatim}
\usepackage{hyphenat,epsfig,subcaption,multirow}
\usepackage{nicefrac}

\usepackage{mathtools}
\usepackage{xparse}

\DeclareFontFamily{U}{mathx}{\hyphenchar\font45}
\DeclareFontShape{U}{mathx}{m}{n}{
      <5> <6> <7> <8> <9> <10>
      <10.95> <12> <14.4> <17.28> <20.74> <24.88>
      mathx10
      }{}
\DeclareSymbolFont{mathx}{U}{mathx}{m}{n}

\usepackage{tcolorbox}
\tcbuselibrary{skins,breakable}
\tcbset{enhanced jigsaw}

\usepackage[normalem]{ulem}

\definecolor{DarkRed}{rgb}{0.5,0.1,0.1}
\definecolor{DarkBlue}{rgb}{0.1,0.1,0.5}

\usepackage{nameref}
\definecolor{ForestGreen}{rgb}{0.1333,0.5451,0.1333}
\definecolor{Red}{rgb}{0.9,0,0}
\usepackage[capitalise,noabbrev,nameinlink]{cleveref}
\crefname{property}{property}{Property}
\creflabelformat{property}{(#1)#2#3}
\crefname{equation}{eq}{Eq}
\creflabelformat{equation}{(#1)#2#3}

\usepackage{bm}
\usepackage{url}
\usepackage{xspace}
\usepackage[mathscr]{euscript}

\usepackage{tikz}
\usetikzlibrary{arrows}
\usetikzlibrary{arrows.meta}
\usetikzlibrary{shapes}
\usetikzlibrary{backgrounds}
\usetikzlibrary{positioning}
\usetikzlibrary{decorations.markings}
\usetikzlibrary{patterns}
\usetikzlibrary{calc}
\usetikzlibrary{fit}
\usetikzlibrary{snakes}
\tikzset{vertex/.style={circle, black, fill=Yellow, line width=1pt, draw, minimum width=8pt, minimum height=8pt, inner sep=0pt}}
	
\usepackage[framemethod=TikZ]{mdframed}

\usepackage[noend]{algpseudocode}
\makeatletter
\def\BState{\State\hskip-\ALG@thistlm}
\makeatother

\crefname{fact}{Fact}{Facts}

\ThCSnewtheoita{result}

\ThCSnewtheostd{Algorithm}

\newenvironment{Ourbox}{\begin{mdframed}[hidealllines=false,innerleftmargin=10pt,backgroundcolor=white!10,innertopmargin=10pt,innerbottommargin=10pt,roundcorner=10pt]}{\end{mdframed}}

\allowdisplaybreaks

\renewcommand{\leq}{\leqslant}
\renewcommand{\geq}{\geqslant}

\DeclarePairedDelimiter{\bracket}[]
\DeclarePairedDelimiter{\paren}()

\DeclarePairedDelimiter{\card}{\vert}{\vert}

\DeclarePairedDelimiter{\set}{\{}{\}}

\DeclareMathOperator*{\distrib}{dist}
\DeclareMathOperator*{\support}{supp}
\DeclareMathOperator*{\kldiv}{\mathbb{D}}
\DeclareMathOperator*{\entropy}{\mathbb{H}}
\DeclareMathOperator*{\inform}{\mathbb{I}}
\DeclareMathOperator*{\expect}{\mathbb{E}}
\DeclareMathOperator*{\variance}{Var}
\DeclareMathOperator*{\covariance}{CoVar}

\newcommand{\Ot}{\widetilde{O}}
\DeclarePairedDelimiterXPP{\Omgt}[1]{\widetilde{\Omega}}(){}{#1}
\DeclarePairedDelimiterXPP{\BigO}[1]{O}(){}{#1}

\NewDocumentCommand{\Prob}{sO{}E{_}{{}}m}{%
  \Pr_{#3}
  \IfBooleanTF{#1}
  {\bracket*{#4}}
  {\bracket[#2]{#4}}
}

\NewDocumentCommand{\Exp}{sO{}E{_}{{}}m}{%
  \expect_{#3}
  \IfBooleanTF{#1}
  {\bracket*{#4}}
  {\bracket[#2]{#4}}
}

\DeclarePairedDelimiterXPP{\Var}[1]{\variance}[]{}{#1}
\DeclarePairedDelimiterXPP{\Cov}[1]{\covariance}[]{}{#1}
\DeclarePairedDelimiterXPP{\eexp}[1]{\exp}(){}{#1}

\NewDocumentCommand{\Dist}{sO{}E{_}{{}}m}{%
  \distrib_{#3}
  \IfBooleanTF{#1}
  {\paren*{#4}}
  {\paren[#2]{#4}}
}

\DeclarePairedDelimiterXPP{\Supp}[1]{\support}(){}{#1}

\DeclarePairedDelimiterXPP{\KL}[2]{\kldiv}(){}{#1 \;\delimsize\|\; #2}

\DeclarePairedDelimiterXPP{\Ent}[1]{\entropy}(){}{#1}
\DeclarePairedDelimiterXPP{\Inf}[1]{\inform}(){}{#1}

\newcommand{\eps}{\varepsilon}

\newcommand{\poly}{\mbox{\rm poly}}

\newcommand{\IR}{\mathbb{R}}

\newcommand{\IN}{\mathbb{N}}

\newenvironment{tbox}{\begin{tcolorbox}[
		enlarge top by=5pt,
		enlarge bottom by=5pt,
		 breakable,
		 boxsep=0pt,
                  left=4pt,
                  right=4pt,
                  top=10pt,
                  arc=0pt,
                  boxrule=1pt,toprule=1pt,
                  colback=white
                  ]
	}
{\end{tcolorbox}}

\newcommand{\mireal}[1][]{
  \ifx\relax#1\relax%
    \II(\mione \,; \mitwo)%
  \else%
    \II(\mione \,; \mitwo\mid #1)%
  \fi
}

\title{A Simple \texorpdfstring{$(1-\eps)$}{(1-eps)}-Approximation Semi-Streaming Algorithm for Maximum (Weighted) Matching} 
\ThCSauthor[waterloo]{Sepehr Assadi}{sepehr@assadi.info}[0009-0006-8914-5995]
\ThCSaffil[waterloo]{School of Computer Science, University of Waterloo.}
\ThCSthanks{Supported in part by a Sloan Research Fellowship, an NSERC Discovery Grant, and a Faculty of Math Research Chair grant from University of Waterloo.
This article was invited from \emph{Symposium on Simplicity in Algorithms, {SOSA} 2024} \cite{Assadi24}.}
\ThCSshortnames{S.\ Assadi}  
\ThCSshorttitle{A Semi-Streaming Algorithm for Maximum Matching}
\ThCSyear{2025}
\ThCSarticlenum{16}
\ThCSreceived{Apr 23, 2024} 
\ThCSrevised{Nov 29, 2024}
\ThCSaccepted{Feb 07, 2025}
\ThCSpublished{Aug 01, 2025}
\ThCSdoicreatedtrue
\ThCSkeywords{Graph streaming algorithms, Semi-streaming algorithms, Multiplicative weight update method}

\begin{document}
\maketitle



\begin{abstract}


	We present a simple semi-streaming algorithm for $(1-\eps)$-approximation of bipartite matching in $O(\log{\!(n)}/\eps)$ passes. This 
	matches the performance of state-of-the-art ``$\eps$-efficient'' algorithms---the ones with much better dependence on  $\eps$ albeit with some mild dependence on $n$---while being considerably simpler. 
	
	
	The algorithm relies on a direct application of the multiplicative weight update method with a self-contained primal-dual analysis that can be of  independent interest. 
	To showcase this, we use the same ideas, alongside standard tools from matching  theory, to present an equally simple semi-streaming algorithm for $(1-\eps)$-approximation of weighted matchings in general (not necessarily bipartite) 
	graphs, again in $O(\log{\!(n)}/\eps)$ passes. 
\end{abstract}






\section{Introduction}\label{sec:intro}

We consider the maximum matching problem in the semi-streaming model of~\cite{FeigenbaumKMSZ05}: given any $n$-vertex graph $G=(V,E)$ whose edges are presented in a stream, the goal is to make a minimal number of passes
over this stream and use a limited space of $\Ot(n):= O(n \cdot \poly\!\log{(n)})$ bits to output a $(1-\eps)$-approximate maximum matching of $G$ for some given $\eps > 0$. 

The maximum matching problem is arguably the most studied problem in the graph streaming literature at this point (see, e.g.~\cite{AssadiS23} for a quick summary). 
Most relevant to our work, the first $(1-\eps)$-approximation algorithm for maximum cardinality matching
was designed by~\cite{McGregor05} which requires $(1/\eps)^{O(1/\eps)}$ passes. 
This algorithm has since been improved numerous times~\cite{EggertKS09,AhnG11,AhnG11b,EggertKMS12,Kapralov13,AhnG15,BehnezhadDETY17,Tirodkar18,GamlathKMS19,AssadiLT21,FischerMU22,AssadiJJST22,HuangS23} culminating in the state-of-the-art, consisting of two incomparable families of algorithms:  

\begin{itemize}[itemsep=5pt]
	\item \textbf{``Constant pass'' algorithms.} We have the algorithm of~\cite{AssadiLT21} (and its precursor~\cite{AhnG11}) with $O(1/\eps^2)$ passes for {bipartite} graphs and that of~\cite{FischerMU22} with $\poly{(1/\eps)}$ passes for general graphs. Similarly, for weighted graphs, 
	we have the algorithm of~\cite{AhnG11} with $O(\log{(1/\eps)}/\eps^2)$ passes for bipartite graphs and~\cite{HuangS23} with $\poly{(1/\eps)}$ passes for general graphs\footnote{It is worth mentioning that the dependence on
	$\eps$ in \cite{FischerMU22, HuangS23} is quite high -- it appears to be $O((1/\eps)^{19})$ passes in~\cite{FischerMU22} (see Lemma 5.6 of arXiv version 5) and can only be higher in~\cite{HuangS23}. \label{footnote:eps}}. 

	\item \textbf{``$\eps$-efficient'' algorithms.\footnote{We use this term to refer to algorithms that have  much better dependence on parameter $\eps$ compared to the above line of work at the cost of having some mild dependence on $n$, which
	is almost always an $O(\log{n})$ factor.}} We have the algorithm of~\cite{AhnG15} (and its precursor~\cite{AhnG11b}) for weighted general graphs with $O(\log{\!(n)}/\eps)$ passes, and a simpler and space-optimal algorithm of~\cite{AssadiJJST22}
	with $O(\log{n}\cdot\log{(1/\eps)}/\eps)$ passes that is specific to bipartite cardinality matching.  
\end{itemize}

See~\Cref{tab:results} for a detailed summary. Nevertheless, despite significant progress in bringing down the pass-complexity of more general cases, for the most basic version of the problem, namely, \emph{maximum (cardinality) bipartite matching (MBM)}, 
the best bounds have been stuck at $O(1/\eps^2)$  and $O(\log{\!(n)}/\eps)$ passes for over a decade now (since~\cite{AhnG11} and~\cite{AhnG11b}, respectively). On the other hand, 
even the recent advances on multi-pass graph streaming lower bounds in~\cite{AssadiR20,ChenKPSSY21,Assadi22,AssadiS23,KonradN24} only rule out $o(\log{(1/\eps)})$-pass algorithms for MBM~\cite{AssadiS23} (under a certain combinatorial hypothesis), leaving an exponential gap open for progress on both ends. 

In our opinion, one contributing factor to the lack of algorithmic progress is the fact that the $O(\log{\!(n)}/\eps)$-pass algorithms of~\cite{AhnG11b,AhnG15} are quite complicated (even for MBM). While some simplifications have been 
made in~\cite{AssadiJJST22}, even this new algorithm is far from being simple. This is in contrast with constant-pass algorithms that, at least for MBM, 
admit quite simple algorithms in~\cite{AssadiLT21} (even already from~\cite{AhnG11}). The goal of this paper is to remedy this state of affairs. 

\begin{table}[h!]

{\smaller
\begin{tabular}{|c | c | c | c | c | c | c |}

 \hline
  ~~~~Citation~~~~ & ~~Space~~ & ~~~~Passes~~~~ & Bip./Gen. & Size/Weight & Det./Rand.  \\ 
  \hline\hline
  \multicolumn{6}{|c|}{``constant-pass'' algorithms} \\
  \hline\hline
  \cite{McGregor05} & $\Ot_{\eps}(n)$ & $(1/\eps)^{O(1/\eps)}$ & Gen. & Size & Rand.  \\
  \hline
  \cite{EggertKS09} & $O(n\log{n})$ & $O(1/\eps^{8})$ & Bip. & Size & Det.  \\
  \hline
  \cite{AhnG11} & $\Ot(n \cdot \poly{(1/\eps)})$ & $O((1/\eps^2) \cdot \log{(1/\eps)})$ & Bip. & Weight & Det. \\ 
  \hline
  \cite{EggertKMS12} & $O(n\log{n})$ & $O(1/\eps^{5})$  & Bip. & Size & Det. \\
  \hline
  \cite{Kapralov13} & $O(n\log{n})$ & $O(1/\eps^2)$ & {\large $\substack{\text{Bip.} \\ \text{(vertex arrival)}}$} & Size & Det. \\
  \hline
  \cite{Tirodkar18} & $\Ot_{\eps}(n)$ & $(1/\eps)^{O(1/\eps)}$ & Gen. & Size & Det. \\
  \hline
  \cite{GamlathKMS19} & $\Ot_{\eps}(n)$ & $(1/\eps)^{O(1/\eps^2)}$ & Gen. & Weight & Det. \\
  \hline
  \cite{AssadiLT21} & $O(n\log{n})$ & $O(1/\eps^2)$ & Bip. & Size & Det. \\
    \hline
  \cite{FischerMU22} & $\Ot(n \cdot \poly{(1/\eps)})$ & $O((1/\eps)^{19})$ & Gen. & Size & Det. \\
  \hline
  \cite{HuangS23} & $\Ot(n \cdot \poly{(1/\eps)})$ & ${(1/\eps)^{O(1)}}$  & Gen. & Weight & Det. \\
  \hline\hline
  \multicolumn{6}{|c|}{``$\eps$-efficient'' algorithms} \\
  \hline\hline
  \cite{AhnG11} & $\Ot(n \cdot \poly{(1/\eps)})$ & $O(\log{n} \cdot \poly{(1/\eps)})$ & Gen. & Weight & Det. \\
\hline
  \cite{AhnG11b} & $\Ot(n \cdot \poly{(1/\eps)})$ & $O(\log{\!(n)}/\eps)$ & Gen. & {\large $\substack{\text{Size} \\ \text{(value not edges)}}$} & Rand. \\
  \hline
  \cite{AhnG15} & $\Ot(n \cdot \poly{(1/\eps)})$ & $O(\log{\!(n)}/\eps)$ & Gen. & Weight & Rand. \\
  \hline
  \cite{AssadiJJST22} & $O(n\log{n})$ & $O(\log{\!(n)}/\eps \cdot \log{(1/\eps)})$ & Bip. & Size & Det. \\
   \hline\hline
  \multicolumn{6}{|c|}{beyond semi-streaming algorithms} \\
  \hline\hline
  \cite{AhnG15} & $\Ot(n^{1+1/p} \cdot \poly{(1/\eps)})$ & $O(p/\eps)$ & Gen. & Weight & Rand. \\
  \hline
  \cite{BehnezhadDETY17} & $O(n^{1.5}\log{\!(n)}/\eps)$ & $O(1/\eps)$ & Bip. & Size & Rand. \\
    \hline
  \cite{AssadiBKL23} & $o_{\eps}(n^2)$ & 1 & Gen. & Size & Rand. \\
   \hline\hline
  \multicolumn{6}{|c|}{multi-pass lower bounds} \\
\hline\hline  
 \cite{ChenKPSSY21} & $n^{1+\Omega(1)}$ & $\Omega({\log{(1/\eps)}}/{\log\log{n}})$ & Bip. & Size & Rand. \\
\hline
\cite{Assadi22} & $n^{1+\Omega(1)}$ & $>2$  & Bip. & Size & Rand. \\
\hline
\cite{AssadiS23} & $n^{1+\Omega(1)}$ & $\Omega(\log{(1/\eps)})$ &Bip. & Size & Rand. \\
\hline
\end{tabular}
}
\caption{\footnotesize Summary of the prior work on $(1-\eps)$-approximate streaming matchings: \emph{Bip./Gen.} refers to bipartite versus general graphs, \emph{Size/Weight} refers to cardinality versus weighted matchings, and \emph{Det./Rand.} refers to deterministic versus randomized algorithms. The space is stated in number of bits and thus $O(n\log{n})$ is space-optimal.\\ 
Due to the elegant ``weighted-to-unweighted'' reduction of~\cite{BernsteinDL21}, for \emph{bipartite} graphs, \emph{all} results for unweighted matchings can be generalized to weighted matchings by increasing the space with a factor of $(1/\eps)^{O(1/\eps)}$, 
while keeping the number of passes \emph{exactly} the same (not only asymptotically). \\ 
\emph{All}  lower bounds stated in the table hold under a combinatorial assumption regarding existence of moderately dense \emph{Ruzsa-Szemeredi} graphs (see~\cite{Assadi22,AssadiS23} for more details; see also~\cite{KonradN24}
for a two-pass unconditional lower bound.). \\
See also~\cite{GoelKK12,Kapralov13,Kapralov21} for much better approximation lower bounds for single-pass semi-streaming algorithms, which currently rule out $1/(1+\ln{2})\approx 0.59$ approximation in~\cite{Kapralov21}. 
Finally, see also~\cite{KhalilK20,KonradN21,KonradNS23} for multi-pass lower bounds that hold for different restricted families of algorithms.
}\label{tab:results}
\end{table}


\subsection*{Our Contributions}
We present a novel way of approximating matchings that is easily implementable via semi-streaming algorithms (among others). 
The high level idea---with some ambiguity left on purpose---is: 
\begin{Ourbox}
	\begin{enumerate}
		\item Sample $\Ot(n/\eps)$ edges uniformly and compute a maximum matching $M$ of the sample. 
		\item If $M$ is large enough, return $M$; otherwise, \emph{(a)} find edges that ``could have potentially led to a larger matching'',
		\emph{(b)} increase their ``importance'',  and repeat the sampling. 
	\end{enumerate}
\end{Ourbox}
This general idea of ``sample-and-solve'' is a staple in the graph streaming literature dating back, at the very least, to the \emph{filtering}~\cite{LattanziMSV11} and \emph{sample-and-prune}~\cite{KumarMVV13} techniques (both used for implementing greedy algorithms). It relies on a fundamental power of semi-streaming algorithms: once we \emph{sparsify} the input to fit into the memory, we can process it however we want. Specifically in this context, 
once the algorithm only has $\Ot(n/\eps)$ edges to work with in the sample, it can find its maximum matching or perform any ``heavy'' computation easily. 

The approach proposed above, based on adjusting importance of edges, is clearly reminiscent of the \emph{Multiplicative Weight Update (MWU)} method (see~\cite{AroraHK12}) and its application in the Plotkin-Shmoys-Tardos framework for approximating
packing/covering LPs~\cite{PlotkinST91}. There is \emph{just} one issue here: we would like this algorithm to converge in $\approx 1/\eps$ passes, while these MWU-based approaches
only guarantee $\approx 1/\eps^2$ iterations for convergence to a $(1-\eps)$-approximate solution (see, e.g.,~\cite{KleinY99}). Addressing this issue is the key difference in our work compared to prior work. 

\paragraph{Prior approaches.} The algorithms in~\cite{AhnG11b,AhnG15} also start with the same overall approach and address the above-mentioned issue through several steps: 
$(i)$ using a non-standard LP relaxation of the problem, $(ii)$ relying on the dual variables of this LP to guide step \emph{(a)} of the approach, $(iii)$ adding a penalty-term to the LP to maintain an $O(\log{\!(n)}/\eps^2)$-iterations
convergence guarantee as in the Plotkin-Shmoys-Tardos framework (to reduce the \emph{width} of the resulting problem; see~\cite{PlotkinST91,AroraHK12}), $(iv)$ ``folding'' $O(1/\eps)$ iterations of this framework in $O(1)$ passes, and $(v)$ using a notion of a 
``deferred (cut) sparsification'' (instead of sampling) that allows for implementing this last step. We refer the reader to \cite[Section 1 and 2 of arXiv version 3]{AhnG15} for more details on this algorithm; here, we only note
that the end result is a highly sophisticated algorithm that barely resembles the above strategy but can now run in $O(\log{\!(n)}/\eps)$ passes. 

The  recent algorithm of~\cite{AssadiJJST22} deviates from the above approach. Instead, it relies on more sophisticated optimizations tools in~\cite{Sherman17,JambulapatiST19,CohenST21} 
based on the classical work in~\cite{Nemirovski04,Nesterov07} that give $\approx 1/\eps$-iteration solvers directly. This approach, to a certain degree, does not take advantage of the aforementioned power of semi-streaming algorithms---meaning arbitrary computation power on sparse-enough inputs---and
is highly tailored to bipartite cardinality matching\footnote{The work of~\cite{AssadiJJST22} also have an algorithm for weighted bipartite matching but the pass-complexity depends linearly on the maximum weight of an edge (which can be polynomial in $n$) and hence is typically not  efficient.}.

\paragraph{Our approach.} Unlike prior work, we are going to revert to the original approaches of~\cite{LattanziMSV11,KumarMVV13} and 
implement the above algorithmic approach \emph{quite literally}, without relying on Plotkin-Shmoys-Tardos or similar generic frameworks. 

Concretely, our algorithm for MBM is this: in step $(a)$, find 
a \emph{minimum (bipartite) vertex cover} of the \emph{sampled} graph (relying on Konig's theorem; see~\Cref{fact:konig}); we then consider any edge of the \emph{original} graph not covered by this vertex cover
as an edge that ``could have potentially led to a larger matching''. For step \emph{(b)}, we double the importance of these edges (making
them twice as likely to be sampled next)\footnote{As a side note, this is a much more aggressive update rule compared to a typical MWU application, say in Plotkin-Shmoys-Tardos framework, which would have updated the weights
 by only a $(1 + \eps)$ factor; see also~\Cref{sec:MWU}.}. A simple analysis, relying on the duality of matchings and vertex covers,  bounds 
the number of iterations by $O(\log{\!(n)}/\eps)$, leading to the following result. 

\begin{result}\label{res:MBM}
\emph{There is a semi-streaming algorithm that given any $n$-vertex bipartite graph $G$ and a parameter $\eps \in (0,1)$, uses $O(n\log{\!(n)}/\eps)$ bits of space and $O(\log{\!(n)}/\eps)$ passes
	and with exponentially high probability\footnote{Here, and throughout, `with exponentially high probability' means with probability at least $1-\exp\paren{-\Theta(n)}$.} outputs a $(1-\eps)$-approximate maximum matching of $G$.} 
\end{result}

We believe~\Cref{res:MBM} is our main contribution as it already contains our key new ideas. Still, 
this result can be significantly generalized to cover maximum weight matchings in general graphs, with minimal additional effort and by relying on standard tools from matching theory. 

\begin{result}\label{res:MWM}
\emph{There is a semi-streaming algorithm that given any $n$-vertex general graph $G$ with integer edge weights $w: E \rightarrow \IN$ and a parameter $\eps \in (0,1)$, uses $O(n\log^2{\!(n)}/\eps)$ bits of space and $O(\log{\!(n)}/\eps)$ passes
	and with exponentially high probability outputs a $(1-\eps)$-approximate maximum weight matching of $G$.} 
\end{result}
\Cref{res:MWM} is now providing a considerably simpler algorithm compared to the main results of~\cite{AhnG15} (also with a better space-dependence by $\poly{(\log{n},1/\eps)}$ factors). We hope this can pave the way for both future theoretical improvements  
and more practical algorithms for this fundamental problem\footnote{It is worth mentioning that the generic approaches of~\cite{LattanziMSV11,KumarMVV13} that are most similar to our algorithms have 
indeed led to highly practical algorithms; see the empirical evaluations in the aforementioned papers.}. 

Finally, we note that in the interest of keeping the main ideas in this paper as transparent as possible, we have opted to focus only on the most important aspects of our algorithms 
in our main arguments. Then, in~\Cref{sec:extensions}, we point out several standard and not-so-standard extensions of our algorithms such as improved runtime, $O(1/\eps)$-pass algorithms in $n^{1+\Omega(1)}$ space, derandomization,  and others.


\newcommand{\qe}[1]{\ensuremath{q^{(#1)}}}
\newcommand{\Qe}[1]{\ensuremath{Q^{(#1)}}}

\newcommand{\pe}[1]{\ensuremath{p^{(#1)}}}

\newcommand{\UU}{\ensuremath{\mathcal{U}}}

\section{Preliminaries}\label{sec:prelim}

\paragraph{Notation.} For any graph $G=(V,E)$, we use $n$ to denote the number of vertices and $m$ as the number of edges. 
We further use $\mu(G)$ to denote the maximum matching size in $G$ and $\mu(G,w)$ to denote the maximum matching weight in $G$ under edge weights $w: E \rightarrow \IN$.

In the following, we review some basic facts from matching theory in bipartite and general (weighted) graphs, separately. We note that our~\Cref{res:MBM} and the arguments in~\Cref{sec:bipartite} only 
rely on the basics for bipartite graphs, and thus a reader solely interested in that part of the paper, can safely skip~\Cref{sec:basics-general} below. 

\subsection{Basics of Matching Theory in Bipartite Graphs}\label{sec:basics-bipartite}

Let $G=(L,R,E)$ be a bipartite graph. Recall the following definitions: 
\begin{itemize}
\item A matching $M$ is a set of vertex-disjoint edges in $E$ and a fractional matching $x \in [0,1]^E$ is an assignment to the edges 
so that for every vertex $v \in L \cup R$, we have $\sum_{e \ni v} x_e \leq 1$. We denote the size of a fractional matching $x$ by $\card{x} := \sum_{e \in E} x_e$. 
\item Similarly, a vertex cover $U$ is a set of vertices incident on every edge and a fractional vertex cover $y \in [0,1]^V$ is an assignment to the vertices so that for every edge $e=(u,v) \in E$, $y_u + y_v \geq 1$. 
We denote the size of a fractional vertex cover $y$ by $\card{y} := \sum_{v \in L \cup R} y_v$. 
\end{itemize}
(We only use fractional matchings and vertex covers  in the analysis). 

Konig's theorem~\cite{Konig1931} establishes duality of maximum (fractional) matchings and minimum (fractional) vertex covers  in bipartite graphs. 
\begin{fact}[{Konig's theorem}]\label{fact:konig}
In bipartite graphs, the sizes of maximum matchings, fractional matchings, minimum vertex covers, and fractional vertex covers are all the same. 
\end{fact}

See the excellent book of Lov\'asz and Plummer~\cite{LovaszP09} on matching theory for more details. 

\subsection{Basics of Matching Theory in General (Weighted) Graphs}\label{sec:basics-general}

Let $G=(V,E)$ be a (general) graph with \emph{integer} edge weights $w: E \rightarrow \IN$. The duality between matchings and vertex covers no longer holds in general graphs, nor 
the equivalence of fractional matchings and integral ones (the way defined previously). Thus, one needs a more general definition. In the following, we use $odd(V)$ to denote
the collection of all sets of vertices in $V$ with \emph{odd} cardinality. For a set $S \subseteq V$, we use $E[S]$ to denote the edges with \emph{both} endpoints in $S$. 

\begin{itemize}
\item As before, a matching $M$ is a set of vertex-disjoint edges in $E$ and $w(M)$ is its weight. We define a (general) fractional matching $x \in [0,1]^E$ as an assignment to the edges satisfying: 
\begin{alignat*}{2}
	&\text{for all $v \in V$:} ~~\sum_{e \ni v} x_e \leq 1 \quad \text{and for all $S \in odd(V)$:} ~~ \sum_{e \in E[S]} x_e \leq \frac{\card{S}-1}{2}. 
\end{alignat*}
We define the weight of a fractional matching as $\sum_{e \in E} w(e) \cdot x_e$. 

\item The dual to maximum fractional matchings is the following \emph{``odd-set cover''} problem. We define a fractional odd-set cover as a pair of assignments $y \in \IR^V$ and $z \in \IR^{odd(V)}$ to vertices and odd-sets in $G$ 
satisfying the following: 
\begin{align}
	\text{for all $e=(u,v) \in E$:} \quad y_u + y_v + \hspace{-10pt}\sum_{\substack{S \in odd(V) \\ e \in E[S]}} \hspace{-10pt} z_S \geq w(e). \label{eq:fractional-odd-set-cover}
\end{align}
The value of a fractional odd-set cover is then
\[
	\card{(y,z)} := \sum_{v \in V} y_v + \sum_{S \in odd(V)} \frac{\card{S}-1}{2} \cdot z_S. 
\]
\end{itemize}

We have the following result---similar in spirit to~\Cref{fact:konig} for bipartite graphs---on the duality of maximum fractional matchings and odd-set covers. 

\begin{fact}[{Edmond's matching polytope theorem~\cite{Edmonds65}}]\label{fact:edmonds}
In any graph, the weights of maximum weight matchings and fractional matchings, and the sizes of minimum odd-set covers are the same. 
\end{fact}

Finally, we shall also use the following result on the structure of optimal fractional odd-set covers, which is crucial for the probabilistic analysis of our algorithm. 
Recall that a family of sets $\mathcal{F}$ is \emph{laminar} iff for all $A,B \in \mathcal{F}$, either $A \cap B = \emptyset$ or $A \subseteq B$ or $B \subseteq A$.

\begin{fact}[{Cunningham-Marsh theorem~\cite{CunninghamM78}}]\label{fact:cunningham-marsh}
	In any graph $G$ with integer edge-weights, there is an optimal fractional odd-set cover $y \in \IR^V$ and $z \in \IR^{odd(V)}$ such that $(i)$ both $y$ and $z$ only take integer values, 
	and $(ii)$  $z_S > 0$ only for a family of sets in $odd(V)$ that form a laminar family. 
\end{fact}

Again, see~\cite{LovaszP09} for more details and proofs of these facts. 

\section{Maximum Cardinality Bipartite Matching}\label{sec:bipartite}

We prove~\Cref{res:MBM} in this section.  We start with presenting our new algorithm in a generic and model-independent way, and the show how it can be implemented in the semi-streaming model. 

\begin{Algorithm}\label{alg:bipartite}
A sample-and-solve approximation algorithm for maximum bipartite matching. 

\begin{itemize}
	\item \textbf{Input}: A bipartite graph $G=(L,R,E)$ and parameter $\eps \in (0,1)$; 
	\item \textbf{Output}: A $(1-\eps)$-approximate maximum matching in $G$. 
\end{itemize}
	\begin{enumerate}
	\item Start with \textbf{importance}\footnote{While it is more common to refer to this concept as `weight' in the context of MWU, given that we will eventually work with weighted matchings, we use
	the term `importance' to avoid ambiguity.} $\qe{1}_e = 1$ for every edge $e \in E$ and define $\Qe{1} := \sum_{e \in E} \qe{1}_e$. 

	\item For $r=1$ to $R := \dfrac{4}{\eps} \cdot \log{m}$ iterations: 
	\begin{enumerate}[leftmargin=5pt]
	\item Sample each edge $e \in E$ independently at random with probability:
	\begin{align}
		p^{(r)}_e := {\frac{2n}{\eps} \cdot \frac{\qe{r}_e}{\Qe{r}}}. \label{eq:pe-bipartite}
	\end{align}
	\item\label{line:2} Compute a maximum matching $M^{(r)}$ and a minimum vertex cover $U^{(r)}$ of the sample. 
	\item For any edge $e \in E$ \underline{not} covered by $U^{(r)}$, update: 
	\begin{align}
		\qe{r+1}_e = \qe{r}_e \cdot 2. \label{eq:qe-bipartite}
	\end{align}
	Then, let $\Qe{r+1} = \sum_{e \in E} \qe{r+1}_e$. 
	\end{enumerate}
	\item Return the largest of matchings $M^{(r)}$ for $r \in [R]$. 

	\end{enumerate}
\end{Algorithm}

\begin{theorem}\label{thm:alg-bipartite}
	For any bipartite graph $G=(L,R,E)$ and parameter $\eps \in (0,1)$, \Cref{alg:bipartite} outputs a matching of size at least $(1-\eps) \cdot \mu(G)$ in $G$ with exponentially high probability. 
\end{theorem}

The proof follows the recipe of the MWU analysis (e.g., in~\cite{AroraHK12}), using $\Qe{r}$ as a potential function. The first lemma upper bounds $\Qe{R+1}$ at the end of the algorithm. 
\begin{lemma}\label{lem:Qr-small-bipartite}
	With exponentially high probability, $\Qe{R+1} \leq (1+\eps/2)^{R} \cdot m$.  
\end{lemma}
\begin{proof}
	Fix any iteration $r \in [R]$. Let $F^{(r)} \subseteq E$ denote the set of edges {not} covered by $U^{(r)}$. We claim that with exponentially high probability, we have,
	\begin{align}
		\sum_{e \in F^{(r)}} \qe{r}_e \leq \frac{\eps}{2} \cdot \Qe{r}. \label{eq:clm1-bipartite}
	\end{align}
	In words,~\Cref{eq:clm1-bipartite} states that the importance of edges {not} covered by $U^{(r)}$ in the graph, relative to the importances of iteration $r$, is ``small'' (despite the fact that $U^{(r)}$ was computed
	on a sample and not the entire input). Before proving this claim, let us see how it concludes the proof. 
	
	By a union bound over $O(\log{(n)}/\eps)$ iterations of the algorithm, we have that for every $r \in [R]$, 
	\begin{align*}
		\Qe{r+1} &= \sum_{e \in F^{(r)}} \qe{r+1}_e + \sum_{e \in E \setminus F^{(r)}} \qe{r+1}_e \tag{by partitioning the edges in and out of $F^{(r)}$} \\
		&= \sum_{e \in F^{(r)}} 2 \cdot \qe{r}_e + \sum_{e \in E \setminus F^{(r)}} \qe{r}_e \tag{by the update rule of the algorithm} \\
		&= \left(\sum_{e \in F^{(r)}} \qe{r}_e\right) + \Qe{r} \tag{by the definition of $\Qe{r}$} \\
		&\leq \paren{1+\frac{\eps}{2}} \cdot \Qe{r},
	\end{align*}
	where the inequality is by~\Cref{eq:clm1-bipartite}. The lemma then follows from this and since $\Qe{1} = m$. 
	
	\noindent
	\emph{Proof of~\Cref{eq:clm1-bipartite}.} Let $U \subseteq V$ be any subset of vertices in the graph and $F(U)$ be the set of edges not covered by $U$. 
	Suppose 
	\[
		\sum_{e \in F(U)} \qe{r}_e > \frac{\eps}{2} \cdot \Qe{r}.
	\]
	We show that in this case, with an exponentially high probability, $U$ cannot be a vertex cover of the sampled edges either. Indeed, we have, 
	\begin{align*}
		\Pr\paren{\text{$U$ is a vertex cover of sampled edges}} &= \Pr\paren{\text{no edge from $F(U)$ is sampled}} \\
		&= \prod_{e \in F(U)} (1-p^{(r)}_e) \tag{by the independence and the sampling probability of edges} \\
		&\leq \exp\paren{-\!\!\sum_{e \in F(U)} p^{(r)}_e} \tag{as $(1-x) \leq e^{-x}$ for all $x \in [0,1]$} \\
		&= \exp\paren{-\frac{2n}{\eps} \cdot \sum_{e \in F(U)} \frac{\qe{r}_e}{\Qe{r}}} \tag{by the choice of $p^{(r)}_e$ in~\Cref{eq:pe-bipartite}} \\
		&\leq \exp\paren{-n}
	\end{align*}
	by our assumption about $U$ and importance of edges in $F(U)$ earlier. 
	
	A union bound over the $2^{n}$ choices of $U$ ensures that with exponentially high probability, any choice of  $U^{(r)}$ that is returned as a vertex cover 
	of sampled edges should satisfy~\Cref{eq:clm1-bipartite}. 
\end{proof}
	
On the other hand, we are going to show that if none of the matchings $M^{(r)}$ is sufficiently large, then the importance of at least one edge should have dramatically increased to the point that it will contradict the bounds in~\Cref{lem:Qr-small-bipartite}. 
The proof of this lemma is based on a simple primal-dual analysis using Konig's theorem in~\Cref{fact:konig}. 

\begin{lemma}\label{lem:Qe-large-bipartite}
	Suppose in every iteration $r \in [R]$, we have $\card{M^{(r)}} < (1-\eps) \cdot \mu(G)$. Then, 
	there exists at least one edge $e \in E$ such that $\qe{R+1}_e \geq 2^{\eps \cdot R}$.  
\end{lemma}
\begin{proof}
	By Konig's theorem (\Cref{fact:konig}) and the assumption in the lemma statement,  we also have that $\card{U^{(r)}} < (1-\eps) \cdot \mu(G)$ for every $r \in [R]$. Let $M^*$ be a maximum matching of $G$. 
	We thus have that in each iteration $r \in [R]$, at least $\eps \cdot \mu(G)$ edges of $M^*$ are not covered. By an averaging argument, this implies that there exists at least one edge $e$ in $M^*$ that is not
	covered in at least $\eps \cdot R$ iterations. Hence, by the update rule of the algorithm, the importance of this edge increases to at least $2^{\eps \cdot R}$, concluding the proof. 
\end{proof}

We are now ready to conclude the proof of~\Cref{thm:alg-bipartite}. 

\begin{proof}[Proof of Theorem~\ref{thm:alg-bipartite}]
	Assume the event of~\Cref{lem:Qr-small-bipartite} happens, thus
	\begin{align*}
		\Qe{R+1} \leq (1+\eps/2)^{R} \cdot m < 2^{3\eps R/4 + \log{m}}. \tag{as $(1+x) < 2^{3x/2}$ for $x > 0$} 
	\end{align*}
	Suppose towards a contradiction that none of the matchings $M^{(r)}$ for $r \in [R]$ computed by the algorithm
	are of size at least $(1-\eps) \cdot \mu(G)$. Then, by~\Cref{lem:Qe-large-bipartite}, there is an edge $e \in E$ such that 
	\[
		\qe{R+1}_e \geq 2^{\eps R}. 
	\]
	Putting these two equations together, as $\qe{R+1}_e \leq \Qe{R+1}$ (by positivity of importances), we obtain 
	\[
		\eps R < 3\eps R/4 + \log{m}, 
	\]
	which only holds for $R < 4\log{m}/\eps$, contradicting the choice of $R$ in the algorithm. Thus, at least one of the matchings returned by the algorithm is of size $(1-\eps) \cdot \mu(G)$, 
	concluding the proof. 
\end{proof}

\subsection*{Semi-Streaming Implementation}

We now present a semi-streaming implementation of~\Cref{alg:bipartite} in the following lemma. 

\begin{lemma}\label{lem:semi-stream-bipartite}
	\Cref{alg:bipartite} can be implemented in the semi-streaming model with $O(n\log{\!(n)}/\eps)$ bits of memory and $O(\log{\!(n)}/\eps)$ passes. 
\end{lemma}
\begin{proof}
	We implement each iteration of the algorithm in $O(1)$ streaming passes. The main part of the  implementation is to maintain the importance of the edges \emph{implicitly}. We do this as follows: 
	\begin{itemize}
	\item For every vertex $v \in V$ and iteration $r \in [R]$, 
	we maintain a  bit $b(v,r)$ denoting if $v$ belongs to the vertex cover $U^{(r)}$ of iteration $r$ (this needs $R = O(\log\!{(n)}/\eps)$ bits per vertex in total);
	\item Whenever an edge $e=(u,v) \in E$ arrives in the stream in the $r$-th pass, we can compute the number of times $e$ has remained uncovered by $U^{(r')}$ for $r' < r$, denoted by $c(e,r)$. This is done by checking $b(u,r')$ and $b(v,r')$ stored
	so far; in particular, 
	\[
	c(e,r) = \card{\set{r' < r \mid b(u,r') = b(v,r') = 0}}.
	\]
	 The importance $\qe{r}_e$ of the edge $e$ in this pass $r$ is then $2^{c(e,r)}$. 
	\item Given we can calculate the importance of each edge upon its arrival, we can first make a single pass and compute the normalization factor $\Qe{r} = \sum_{e \in E} \qe{r}_e$. Then, 
	we make another pass and for each arriving edge $e \in E$, we compute $\qe{r}_e$ as above and sample the edges with the probability prescribed by~\Cref{eq:pe-bipartite}. 
	\end{itemize} 
	
	The rest of the algorithm can be implemented directly. In particular, the total number of edges sampled in each iteration is $O(n/\eps)$ with exponentially high probability by Chernoff bound. Thus, in the semi-streaming algorithm, 
	we can store these edges and then compute $M^{(r)}$ and $U^{(r)}$ at the end of the pass. This concludes the proof of the lemma. 
\end{proof}

\Cref{res:MBM} now follows immediately from~\Cref{thm:alg-bipartite} and~\Cref{lem:semi-stream-bipartite}.



\section{Maximum Weight General Matching}\label{sec:general}

We now switch to proving~\Cref{res:MWM} which is a vast generalization of~\Cref{res:MBM}. Interestingly however, 
despite its generality, the proof is more or less a direct ``pattern matching'' of the previous ideas to general weighted graphs using 
the existing rich theory of matchings reviewed in~\Cref{sec:basics-general}. As before, we start by presenting a model-independent algorithm first followed by its semi-streaming implementation. Also, for simplicity of exposition, we are going to present our algorithm 
with space- and pass-complexity depending on the parameter $W := \sum_{e \in E} w(e)$, and then show how to fix this using
standard ideas and conclude the proof of~\Cref{res:MWM}. Our algorithm is as follows.

\begin{Algorithm}\label{alg:general}
A sample-and-solve approximation algorithm for weighted general matching. 

\begin{itemize}
	\item \textbf{Input}: A (general) graph $G=(V,E)$ with weights $w: E \rightarrow \IN$ and parameter $\eps \in (0,1)$; 
	\item \textbf{Output}: A $(1-\eps)$-approximate maximum weight matching in $G$. 
\end{itemize}

	\begin{enumerate}
	\item Start with {importance} $\qe{1}_e = 1$ for every edge $e \in E$ and define $\Qe{1} := \sum_{e \in E} w(e) \cdot \qe{1}_e$. 
	\item Let $W:= \sum_{e \in E} w(e)$. For $r=1$ to $R := \dfrac{4}{\eps} \cdot \log{W}$ iterations: 
	\begin{enumerate}
	\item Sample each edge $e \in E$ independently at random with probability:
	\begin{align}
		\pe{r}_e := {\frac{8n \cdot \ln{(nW)}}{\eps} \cdot \frac{\qe{r}_e \cdot w(e)}{\Qe{r}}}. \label{eq:pe-general}
	\end{align}
	\item\label{line:2-general} Compute a maximum weight matching $M^{(r)}$ and a minimum odd-set cover solution $(y^{(r)},z^{(r)})$ of the sample (using the original weights $w(\cdot)$ on sampled edges). 
	\item\label{line:we-defined-this} For any edge $e \in E$ \underline{not} covered by $(y^{(r)},z^{(r)})$\footnote{By this, we mean $(y,z)$ does \emph{not} satisfy~\Cref{eq:fractional-odd-set-cover} for the given edge $e \in E$.}, update: 
	\begin{align}
		\qe{r+1}_e = \qe{r}_e \cdot 2. \label{eq:qe-general}
	\end{align}
	Then, let $\Qe{r+1} = \sum_{e \in E} w(e) \cdot \qe{r+1}_e$. 
	\end{enumerate}
	\item Return the maximum weight matching among $M^{(r)}$'s for $r \in [R]$. 
	\end{enumerate}
\end{Algorithm}

\begin{theorem}\label{thm:alg-general}
	For any graph $G=(V,E)$ with weights $w: E \rightarrow \IN$ and parameter $\eps \in (0,1)$, \Cref{alg:general} outputs a matching of weight at least $(1-\eps) \cdot \mu(G,w)$ with exponentially high probability. 
\end{theorem}

We follow the same exact strategy as before. The first step is to bound the total sum of importances across the iterations. The following lemma
is an analogue of~\Cref{lem:Qr-small-bipartite}. The key difference is a new union bound argument at the very end
for all potential odd-set covers which needs to be more careful compared to the trivial $2^n$-bound for vertex covers. 

\begin{lemma}\label{lem:Qr-small-general}
	With an exponentially high probability, $\Qe{R+1} \leq (1+\eps/2)^{R} \cdot W$.   
\end{lemma}
\begin{proof}
	Fix any iteration $r \in [R]$. Let $F^{(r)} \subseteq E$ denote the set of edges \emph{not} covered by $(y^{(r)},z^{(r)})$ as defined in Line~\eqref{line:we-defined-this} of~\Cref{alg:general}. We claim that with exponentially high probability, we have, 
	\begin{align}
		\sum_{e \in F^{(r)}} \qe{r}_e \cdot w(e) \leq \frac{\eps}{2} \cdot \Qe{r}. \label{eq:clm1-general}
	\end{align}
	In words,~\Cref{eq:clm1-general} states that the total ``importance $\times$ weight'' of the edges {not} covered by the odd-set cover solution on the entire graph, relative to the importances of iteration $r$, is ``small''. Before proving this claim, let us see how it concludes the proof. 
	
	By a union bound over all iterations of the algorithm, we have that for every $r \in [R]$, 
	\begin{align*}
		\Qe{r+1} &= \sum_{e \in F^{(r)}} \qe{r+1}_e \cdot w(e) + \sum_{e \in E \setminus F^{(r)}} \qe{r+1}_e \cdot w(e) \tag{by partitioning the edges in and out of $F^{(r)}$} \\
		&= \sum_{e \in F^{(r)}} 2 \cdot \qe{r}_e \cdot w(e) + \sum_{e \in E \setminus F^{(r)}} \qe{r}_e \cdot w(e) \tag{by the update rule of the algorithm} \\
		&= \left({\sum_{e \in F^{(r)}} \qe{r}_e \cdot w(e)}\right) + \Qe{r} \tag{by the definition of $\Qe{r}$} \\
		&\leq \paren{1+\frac{\eps}{2}} \cdot \Qe{r},
	\end{align*}
	where the inequality is by~\Cref{eq:clm1-general}. The lemma then follows from this and the choice of $\qe{1}_e=1$ for all edge $e \in E$ which implies $\Qe{1} = \sum_{e \in E} 1 \cdot w(e) = W$. 
	
	\noindent
	\emph{Proof of~\Cref{eq:clm1-general}.} Let $y \in \IR^V$ and $z \in \IR^{odd(V)}$ be any ``potential'' odd-set cover of $G$. 
	We define $F(y,z) \subseteq E$ as the set of edges \emph{not} covered by this potential odd-cut cover. 
	Suppose 
	\[
		\sum_{e \in F(y,z)} \qe{r}_e \cdot w(e) > \frac{\eps}{2} \cdot \Qe{r};
	\]
	we show that in this case, with an exponentially high probability, $(y,z)$ cannot be a feasible odd-set cover  of the sampled edges either. Indeed, we have, 
	\begin{align*}
		\Pr\paren{\text{$(y,z)$ is feasible on sampled edges}} &= \Pr\paren{\text{no edge from $F(y,z)$ is sampled}} \\
		&= \prod_{e \in F(y,z)} (1-p^{(r)}_e) \tag{by the independence and the sampling probability of edges} \\
		&\leq \exp\paren{-\!\!\sum_{e \in F(y,z)} p^{(r)}_e} \tag{as $(1-x) \leq e^{-x}$ for all $x \in [0,1]$} \\
		&= \exp\paren{-\frac{8n \cdot \ln{(nW)}}{\eps} \cdot \sum_{e \in F(y,z)} \frac{\qe{r} \cdot w(e)}{\Qe{r}}} \tag{by the choice of $p^{(r)}_e$ in~\Cref{eq:pe-general}} \\
		&\leq \exp\paren{-4n \cdot \ln{(nW)}},
	\end{align*}
	by our assumption about $(y,z)$ and importance of edges in $F(y,z)$ earlier. 

	 The last step of the proof is to union bound over all potential odd-set covers $(y,z)$ using the calculated probabilities above. 
	This step needs to be more careful compared to~\Cref{lem:Qr-small-bipartite} because $(y,z)$ can be fractional and even for integral values, $z$ can have  an exponential support, leading to a doubly exponential number of choices for it; this 
	is too much for the above probabilities to handle. Both of these are handled using the Cunningham-Marsh theorem (\Cref{fact:cunningham-marsh}): 
		
	\begin{itemize}[leftmargin=10pt]
		\item Firstly, we can assume without loss of generality that $(y,z)$ only take integer values in $[0:W]$ in the optimal solutions computed in~Line~\eqref{line:2-general} of~\Cref{alg:general}. 
		This, for instance, implies that the total number of choices for $y$ that we need to consider is $(W+1)^{n}$. 
		\item Secondly, and more importantly, we can use the laminarity of $\mathcal{F}(z) := \set{S \in odd(V) \mid z_S > 0}$ in the optimal solution. 
		A standard observation about 
		laminar families (over $n$ elements/vertices) is that they can have size at most $2n-1$\footnote{Without loss of generality, assume $\mathcal{F}$ is a \emph{maximal} laminar family on $[n]$. 
		Proof by induction (base case is trivial): 
		maximality ensures that there are two non-empty sets $A,B \in \mathcal{F}$ with $A \cup B = [n]$ and $A \cap B = \emptyset$. By induction, there are at most $2\card{A}-1$ subsets of $A$ in $\mathcal{F}$, at most $2\card{B}-1$ subsets of $B$ in $\mathcal{F}$, 
		and the set $A \cup B$ can also be in $\mathcal{F}$, implying that $\card{\mathcal{F}} \leq (2\card{A}-1) + (2\card{B}-1) + 1 = 2n-1$.}. We can use this to provide a  (crude) upper bound
		the number of choices for $z$:
		\begin{itemize}
			\item Define $T(n)$ as the number of laminar families on $[n]$\footnote{There are more accurate ways of bounding $T(n)$ (see, e.g.~\cite{Bui-XuanHR12}), but given
			that these more accurate bounds do not help with our subsequent calculations, we just establish a crude upper bound with a self-contained proof here.}. We claim $T(n) \leq (4n) \cdot T(n-1)$. We can pick a laminar family on $n-1$ elements
			in $T(n-1)$ ways, and then decide to put the last element in \emph{(a)} one of its (at most) $(2n-3)$ sets and ``propagate'' it to each of its supersets also, \emph{(b)} in a singleton set and decide whether or not to propagate it, or \emph{(c)} not 
			place it anywhere. 
			This leads to $<4n$ options times $T(n-1)$, establishing the claim. 
			\item Given that $T(1) = 2$ (singleton or empty-set), we get $T(n) \leq (4n)^n$. To pick $z$, we can first pick $\mathcal{F}(z)$ using at most $(4n)^n$ ways, and then assign values from $[W]$ to each of its at most $(2n-1)$ sets. Hence,
			there are at most $W^{2n-1} \cdot (4n)^n$ choices for $z$. 
		\end{itemize}
		\end{itemize}
		All in all, we obtain that the total number of possible optimal solutions $(y,z)$ that we need to take a union bound over can be (very) crudely upper bounded by the following (assuming $n > 4$ without loss of generality):   
		\[
			(W+1)^{n} \cdot W^{2n-1} \cdot (4n)^n < (2n)^{2n} \cdot W^{3n} < (nW)^{3n} < \exp\paren{3n \cdot \ln{(nW)}}. 
		\]
		Finally, we can  apply a union bound over these many choices and get that with an exponentially high probability no such solution $(y,z)$ can be a feasible solution on the sample. 
		This concludes the proof. 
\end{proof}
	
We then prove an analogue of~\Cref{lem:Qe-large-bipartite}, using the duality of odd-set covers and (general) matchings, in place of vertex cover/matching duality in bipartite graphs. 

\begin{lemma}\label{lem:Qe-large-general}
	Suppose in every iteration $r \in [R]$, we have $w({M^{(r)}}) < (1-\eps) \cdot \mu(G,w)$. Then, 
	there exists at least one edge $e \in E$ such that $\qe{R+1}_e \cdot w(e)  \geq 2^{\eps \cdot R}$.  
\end{lemma}
\begin{proof}
	By Edmond's matching polytope theorem (\Cref{fact:edmonds}) and the assumption in the lemma statement, we also have that $\card{(y^{(r)},z^{(r)})} < (1-\eps) \cdot \mu(G)$ for every $r \in [R]$. Let $M^*$ be a maximum matching of $G$. 
	We thus have that in each iteration $r \in [R]$, at least $\eps \cdot \mu(G)$ edges of $M^*$ are not covered (as defined in Line~\eqref{line:we-defined-this} of~\Cref{alg:general}); otherwise, another application of  \Cref{fact:edmonds} to 
	the covered edges of $M^*$ implies that $\card{(y^{(r)},z^{(r)})} \geq (1-\eps) \cdot \mu(G)$ also, a contradiction. By an averaging argument, this implies that there exists at least one edge $e$ in $M^*$ that is not
	covered in at least $\eps \cdot R$ iterations. Hence, by the update rule of the algorithm, the importance of this edge increases to at least $2^{\eps \cdot R}$, concluding the proof. 
\end{proof}

We are now ready to conclude the proof of~\Cref{thm:alg-general} exactly as in that of~\Cref{thm:alg-bipartite}, using the above established lemmas instead. 

\begin{proof}[Proof of Theorem~\ref{thm:alg-general}]
	Assume the  event of~\Cref{lem:Qr-small-general} happens, thus
	\begin{align*}
		\Qe{R+1} \leq (1+\eps/2)^{R} \cdot W < 2^{3\eps R/4 + \log{W}}. \tag{as $(1+x) < 2^{3x/2}$ for $x > 0$} 
	\end{align*}
	Suppose towards a contradiction that none of the matchings $M^{(r)}$ for $r \in [R]$ computed by the algorithm
	are of weight at least $(1-\eps) \cdot \mu(G,w)$. By~\Cref{lem:Qe-large-general}, there is an edge $e \in E$ such that 
	\[
		\qe{R+1}_e \geq 2^{\eps R}. 
	\]
	Putting these two equations together, as $\qe{R+1}_e \leq \Qe{R+1}$ (by positivity of importances), we obtain 
	\[
		\eps R < 3\eps R/4 + \log{W}, 
	\]
	which only holds for $R < 4\log{W}/\eps$, contradicting the choice of $R$ in the algorithm. Thus, at least one of the found matchings $M^{(r)}$'s is of weight $(1-\eps) \cdot \mu(G,w)$, 
	concluding the proof.
\end{proof}

\subsection*{Semi-Streaming Implementation}

Finally, we are going to show a semi-streaming implementation of~\Cref{alg:general} in the following lemma. 
The proof is similar to that of~\Cref{lem:semi-stream-bipartite} although again with crucial changes to account for the difference of vertex covers in~\Cref{alg:bipartite} with odd-set covers in~\Cref{alg:general} (we need some minor 
modifications after this also in order to be able to prove~\Cref{res:MWM} which will be done next). 

\begin{lemma}\label{lem:semi-stream-general}
	\Cref{alg:general} is implementable in the semi-streaming model with $\!O(n \log(nW) \cdot \log{\!(n)}/\eps)$ bits of memory and $O(\log{\!(W)}/\eps)$ passes. 
\end{lemma}
\begin{proof}

	We implement each iteration of the algorithm via $O(1)$ passes over the stream. The main part of the semi-streaming implementation is to maintain the importance of the edges implicitly. To do this, we do as follows: 
	\begin{itemize}[leftmargin=10pt]
	\item For every iteration/pass $r \in [R]$: 
	\begin{itemize}
		\item Store the vector $y^{(r)}$ explicitly using $O(n\log{W})$ bits using the integrality of the optimal solution (by~\Cref{fact:cunningham-marsh}). 
		\item Store the vector $z^{(r)}$ explicitly using $O(n\log{n} + n\log{W})$ bits using the laminarity of support of $z$ in the optimal solution (by~\Cref{fact:cunningham-marsh}): this can be done
		because earlier \emph{(a)} we bounded the number of laminar families by $(4n)^n$ and hence they require $O(n\log{n})$ bits of representation, and \emph{(b)} we bounded the number of sets in any laminar family
		by $2n-1$, so we only need to store $O(n\log{W})$ bits to store their values. 
	\end{itemize}
	\item Whenever an edge $e=(u,v) \in E$ arrives in the stream in the $r$-th pass, we can compute the number of times $e$ has remained uncovered by $U^{(r')}$ for $r' < r$, denoted by $c(e,r)$, by checking $(y^{(r')},z^{(r')})$ stored
	so far and counting the number of times they violate~\Cref{eq:fractional-odd-set-cover} for this particular edge $e$. The importance $\qe{r}_e$ of the edge $e$ is then $2^{c(e,r)}$. 
	\item Once we calculate the importance of an edge upon its arrival, we can sample the edge with probability prescribed by~\Cref{eq:pe-general}, by making an additional pass to compute the normalization factor $\Qe{r}$ first. 
	\end{itemize} 
	
	The rest of the algorithm can be implemented directly. In particular, the total number of edges sampled in each iteration is $O(n\log\!{(nW)}/\eps)$ with exponentially high probability by Chernoff bound. Thus, in the semi-streaming algorithm, 
	we can store these edges and then compute $M^{(r)}$ and $(y^{(r)},z^{(r)})$ at the end of the pass. This concludes the proof of the lemma. 
\end{proof}

\begin{proof}[Proof of Result~\ref{res:MWM}]
The only remaining part to prove~\Cref{res:MWM} is to remove the dependence on the parameter $W$, which can be done using an entirely standard idea.  

We can use a single pass to find the maximum weight edge $w^*$ and subsequently, ignore all edges with weight less than $(\eps/n) \cdot w^*$ because the total contribution of those edges to the maximum weight matching
is always less than $\eps$-fraction of its weight.  Thus, we can assume that the weights are in $[1: n/\eps]$ from now on.  
Moreover, we can assume without loss of generality that $\eps > 1/n^3$ as otherwise, we can simply store the entire graph in $O(n^2\log{(n/\eps)}) = O(n\log{\!(n)}/\eps)$ bits (consistent with the bounds of~\Cref{res:MWM}) and 
trivially solve the problem in one pass (as a side note, we are only interested in much larger values of $\eps$ anyway). 

This step implies that without loss of generality, when proving~\Cref{res:MWM}, we can assume all edges have integer weights bounded by  $n^{4}$ by re-scaling $\eps$ to some $\Theta(\eps)$. 
Combining this with~\Cref{thm:alg-general} and~\Cref{lem:semi-stream-general} 
now gives an $O(n\log^2\!{(n)}/\eps)$ memory algorithm in $O(\log{\!(n)}/\eps)$ passes. 
\end{proof}


\section{Further Extensions and Discussions}\label{sec:extensions}

As stated earlier, in the interest of having a clear and concise exposition of our results, we opted to focus solely on the most important aspects of our algorithms in~\Cref{res:MBM} and~\Cref{res:MWM}. We now
 present some natural extensions of our results and further discuss the connection between our work and the Plotkin-Shmoys-Tardos Framework~\cite{PlotkinST91} for solving covering/packing LPs via MWU.

\subsection{Fewer Passes in More Space}\label{sec:ext-few-pass}

Our algorithms, similar to that of~\cite{AhnG15}, can be made more pass efficient at the cost of increasing the space, allowing us to prove the following result: 

\begin{corollary}\label{cor:eps-pass} For every  $p \geq 1$ and  $\eps \in (0,1)$, there is a randomized streaming algorithm for $(1-\eps)$-approximation of maximum weight matching
	in $O(n^{1+1/p}  \log^2{\!(n)}/\eps)$ space and $O(p/\eps)$ passes. 
\end{corollary}

This in particular means that if instead of semi-streaming space, we allow the streaming algorithm to use $n^{1+\delta}$ space for any constant $\delta > 0$, then, the number of passes is only $O(1/\eps)$. 
This also implies that even for semi-streaming algorithms, the pass-complexity can be brought down to $O(\log{\!(n)}/(\log\log{\!(n)} \cdot \eps))$ passes by taking $p=(\log{n}/\log\log{n})$ and having $n^{1+1/p} = \Ot(n)$. 

A common technique for proving~\Cref{cor:eps-pass}, from an already-existing semi-streaming algorithm, is ``delayed sampling'' used, e.g., in~\cite{AhnG15,IndykMRUVY17}: 
we can ``fold'' multiple passes of the semi-streaming algorithm into a single-pass of the larger-space algorithm by oversampling the input first, and then do rejection sampling 
(see~\cite[step 2 of Section 1.3]{IndykMRUVY17} for more details). While this approach would work for us also, 
it would require (slightly) more space  (by some $n^{1/p}/\eps$ factors), and more importantly an indirect analysis. 

Instead, one can directly adjust our importance-sampling based approach in~\Cref{alg:general}. 
It simply involves sampling the edges by a factor of $O(n^{1/p})$ more and then  increasing the importance of violated edges in the algorithm, quite aggressively, by a factor of $O(n^{1/p})$ instead of only $2$. 
We now formalize this proof. 

\begin{proof}[Proof of Corollary~\ref{cor:eps-pass}]
Let $\eta$ be a parameter (that will be later chosen to be $n^{1/p}$). We prove this corollary by modifying~\Cref{alg:general} as follows: 
\begin{enumerate}
	\item Increase the sampling rate in \Cref{eq:pe-general} by a factor of $\eta$. 
	\item Increase the importance of any violated edge in \Cref{eq:qe-general} by a factor of $(1+\eta)$ instead. 
\end{enumerate}
The implications of these changes are: 
\begin{enumerate}
	\item The space of the algorithm is increased by an $O(\eta)$ factor because we store a larger sample. 
	\item The upper bound of $(1+\eps/2)^{R} \cdot W$ on the potential function in~\Cref{lem:Qr-small-general} still holds. This is because~\Cref{eq:clm1-general} now holds with $\eps/2$ replaced by $\eps/2\eta$ which
	``cancels out'' the effect of increasing importances by a $(1+\eta)$ factor instead. 
	
	\item On the other hand the lower bound of $2^{\eps R}$ on the potential function in~\Cref{lem:Qe-large-general} simply becomes $(1+\eta)^{\eps \cdot R}$ 
	given this new update rule on the importances. 
\end{enumerate}
	 Thus, by combining the steps above, we get that 
	\[
		\eps R \log{\eta} \leq 3\eps/4 \cdot R + \log{W}
	\]
	which gives us the bound of $R=O(\log{W}/(\log{\eta} \cdot \eps))$. 

Now, recall that in the proof of~\Cref{res:MWM}, we argued that we can take $W$ to be at most $n^{4}$. This implies that by setting $\eta = n^{1/p}$, we get 
\[
	R = O(\log{W}/(\log{\eta} \cdot \eps)) = O(\log{n}/(1/p \cdot \log{n} \cdot \eps)) = O(p/\eps), 
\]
This concludes the proof of~\Cref{cor:eps-pass}.
\end{proof}

Before moving on, we remark that in addition to~\cite{AhnG15} and~\Cref{cor:eps-pass},~\cite{BehnezhadDETY17} also provides an $O(n^{1.5}/\eps)$-space, $O(1/\eps)$-pass algorithm for maximum (cardinality) bipartite matching.   
The algorithm behind our~\Cref{res:MBM} (and its extension in~\Cref{cor:eps-pass}) turns out to be quite similar in hindsight to the algorithm of~\cite{BehnezhadDETY17} that also relies on a sample-and-solve approach
using vertex covers to guide their sampling; however, unlike our approach,~\cite{BehnezhadDETY17} sticks to uniform sampling and does not adjust any importances, which leads to the larger space-complexity of $O(n^{1.5})$ instead 
of (essentially) $n^{1+o(1)}$-space in our work for $O(1/\eps)$-pass algorithms.

\subsection{Derandomization via Cut Sparsifiers}\label{sec:ext-derandomization}

The failure probability of our algorithms in~\Cref{res:MBM} and~\Cref{res:MWM} are exponentially small, which is better than the typical `with high probability bounds' (namely, $1-1/\poly{(n)}$ bounds) in the same context. 
But, in fact, we can fully derandomize the algorithms also at the cost of increasing the space by $\poly(\log{n},1/\eps)$ factors and keeping the number of passes asymptotically the same. 
We state and prove this part only for unweighted graphs, but with more technical work, one can also extend this to weighted graphs (we omit the latter result as it requires too much of a detour). 

\begin{corollary}\label{cor:deterministic} For every  $p \geq 1$ and  $\eps \in (0,1)$, there is a \underline{deterministic} streaming algorithm for $(1-\eps)$-approximation of maximum cardinality matching
	in $\Ot(n^{1+1/p}/\eps^{2})$ space and $O(p/\eps)$ passes. 
\end{corollary}

~\Cref{cor:deterministic} is proven via replacing the sampling step of the algorithm with \emph{cut sparsifiers} of~\cite{BenczurK96}: these are (re-weighted) subgraphs of the input
that preserve weights of cuts to within a $(1\pm \eps)$-approximation while having only $O(n\log{\!(n)}/\eps^2)$ edges. We note that in general cut sparsifiers are not good at preserving large matchings\footnote{There are examples wherein a cut sparsifier of graph with a perfect matching may only have a 
maximum matching of size $(n^{1/2+o(1)})$ edges (e.g., a union of a perfect matching plus $(n^{1/2+o(1)})$ vertices connected to all other vertices).} but appear to be good at preserving ``near feasible'' vertex covers and odd-set covers 
we need for our primal-dual analysis. Finally, it is known how to compute an $\eps$-cut sparsifier in the streaming model using a single pass and $\Ot(n/\eps^2)$ space deterministically using any deterministic static (non-streaming)
algorithm for this problem, say~\cite{BatsonSS09}; see, e.g.,~\cite{McGregor14} for this elegant and quite simple reduction (based on the \emph{merge-and-reduce} technique dating back to the work of~\cite{RajagopalanML98} on quantile estimation).

Before getting to the formal proof, which focuses on~\Cref{res:MWM} (and~\Cref{cor:eps-pass}), 
let us see an intuition for this result by focusing on derandomizing~\Cref{res:MBM} for MBM instead. Suppose in~\Cref{alg:bipartite}, instead of sampling edges proportional to importances, 
we pick a $\Theta(\eps)$-cut sparsifier $H$  of $G$ with edges weighted by the importances. Then, we simply pick a vertex cover of $H$ (ignoring the weights now). We claim that~\Cref{eq:clm1-bipartite} 
still holds. Let $U \subseteq V$ be a ``potential'' vertex cover so that the total importance of edges it does \emph{not} cover is $> (\eps/2) \cdot Q$, 
where $Q$ is the  importance of all edges in $G$ in this iteration. One can show that $H$ needs to contain at least one edge entirely in $V \setminus U$ 
to be able to be a, say, $(\eps/100)$-cut sparsifier of $G$. 
This means
that $U$ could have not been chosen as a vertex cover of $H$. Thus, vertex covers of $H$ satisfy~\Cref{eq:clm1-bipartite} and the rest of the analysis is the same (recall that $H$ is a subgraph of $G$ so
a ``good'' vertex cover always exist). We now formalize the proof. 

\begin{proof}[Proof of Corollary~\ref{cor:deterministic}]
	We focus on proving the result for semi-streaming algorithms; extending this approach to larger-space algorithms with fewer passes is exactly as in~\Cref{cor:eps-pass} and thus we do not repeat the argument. 
	
	The algorithm is the following. Instead of sampling edges in~\Cref{alg:general}, we compute an $(\eps/100)$-cut sparsifier $H$ of $G$ whose edges are weighted by the importances in this iteration. 
Then, we compute an odd-set cover $(y,z)$ of $H$, ignoring all edge weights in this step, and continue exactly as before. Given that a cut sparsifier can be computed in $\Ot(n/\eps^2)$ space in the semi-streaming model (see~\cite{McGregor14}), 
we will obtain the desired deterministic algorithm. 

Recall that the sampling step was only used in the analysis in~\Cref{lem:Qr-small-general} and in particular to establish~\Cref{eq:clm1-general} with an exponentially high probability. We instead show that this new approach
deterministically satisfies~\Cref{eq:clm1-general}. The rest of the proof of~\Cref{cor:deterministic} then follows verbatim as in~\Cref{res:MWM}. Thus, 
we only need to show the following: 
\begin{itemize}
	\item Let $H$ be an $(\eps/100)$-cut sparsifier of $G=(V,E)$ under the edge weights $q_e$. Then, any odd-set cover $(y,z)$ of $H$ satisfies~\Cref{eq:clm1-general} deterministically, i.e., 
	\begin{align}
		\sum_{e \in F(y,z)} q_e \leq (\eps/2) \cdot Q, \label{eq:not-hold}
	\end{align}
	where $F(y,z)$ is the set of violated edges by $(y,z)$ in the unweighted graph $G$, and $q_e$ and $Q$ are the importance of edge $e \in E$, and total importance of all edges, respectively. 
\end{itemize}
We now prove this statement. 
In the following, for a graph $G=(V,E)$ and two disjoint sets of vertices $A,B \subseteq V$, we define $cut_G(A)$ and $cut_G(A,B)$ 
as the weight of the cuts $(A, V \setminus A)$ and $(A,B)$, respectively (we apply this to $G$ with weight function being edge importances, and to $H$ with the re-weighted weights of the sparsifier).  

Let $(y,z)$ be an optimal odd-set cover of $H$. By Cunningham-Marsh theorem (\Cref{fact:cunningham-marsh}), 
$y$ and $z$ are both integral and $\mathcal{F}(z) := \set{S \in odd(V) \mid z_S > 0}$ forms a laminar family. Moreover, given the optimality of $(y,z)$ and since $H$ is unweighted (when calculating the odd-set cover), 
we have that $y,z \in \set{0,1}^n$ which implies that $\mathcal{F}(z) = S_1,S_2,\ldots,S_s$ is actually a collection of disjoint sets, and is disjoint from the support of $y$, denoted by $T$. 
Notice that the set of violated edges by $(y,z)$ in $G$ are the ones
that are not inside $S_1,\ldots,S_s$, nor incident on $T$. 

Suppose towards a contradiction that~\Cref{eq:not-hold} does not hold. Let $(A,B)$ be a maximum cut of the graph $G[V \setminus T]$ among all cuts where each $S_i \in \mathcal{F}(z)$ is entirely on one side of the cut. 
Since a maximum cut always has weight at least half of the weight of edges in the graph, we have, 
\[
	cut_G(A,B) > (\eps/4) \cdot Q.  
\]
On the other hand, in any graph, we also have
\[
	cut_G(A) + cut_G(B) = cut_G(A \cup B) + cut_G(A, B). 
\]
Given that $cut_G(A),cut_G(B),cut_G(A \cup B) \leq Q$ trivially, and since $H$ is an $(\eps/100)$-cut sparsifier, 
\begin{align*}
	cut_H(A,B) &= cut_H(A) + cut_H(B) - cut_H(A \cup B) \\
	&\geq cut_G(A) + cut_G(B) - cut_G(A \cup B) - 3 \cdot (\eps/100) \cdot Q \\
	&= cut_G(A,B) - 3 \cdot (\eps/100) \cdot Q \\
	&\geq (\eps/4) \cdot Q - 3 \cdot (\eps/100) \cdot Q > 0,
\end{align*}
which implies that there is at least one edge between $A$ and $B$ in $H$. But recall that none of the edges between $A$ and $B$ were covered by $(y,z)$, contradicting the fact that $(y,z)$ was a 
feasible odd-set cover of $H$. This proves~\Cref{eq:not-hold}. 
\end{proof}

We shall remark that we were inspired by the use of cut sparsifiers in~\cite{AhnG15} for this part of the argument. Although, to our knowledge, the use of sparsifiers in~\cite{AhnG15} is for a different purpose
of their ``delayed sparsification'' and folding $O(1/\eps)$ iterations of their optimization method in $O(1)$ passes; we are instead using them for derandomization purposes (the algorithms of~\cite{AhnG15} are randomized despite using sparsifiers even
in insertion-only streams).

\subsection{Running Times of Our Algorithms}\label{sec:ext-runtime}

Given that the main resource of interest in the streaming model is the space, we did not put any emphasis on the runtime of our algorithms in the preceding discussions. 
It is clear that our algorithms run in polynomial time since finding maximum matchings and minimum vertex covers in bipartite graphs or minimum odd-set covers in general graphs can all
be done in polynomial time. However, our algorithms can be made more time efficient, captured by the following corollary. 

\begin{corollary}\label{cor:time-efficient}
	Both algorithms in~\Cref{res:MBM} and~\Cref{res:MWM} can be implemented in $\Ot(m/\eps^2 + n/\eps^3)$ time and the same asymptotic space and pass complexity. 
\end{corollary}
\noindent
We note that in the semi-streaming model, our algorithms can only handle values of $\eps$ such that $1/\eps = \poly\log{(n)}$, as otherwise the space of the algorithm will be more than the semi-streaming restriction of $\Ot(n)$ bits. 
In this regime, our algorithms run in near-linear time. 

We first recall the seminal algorithm of~\cite{DuanP14}---generalizing classical results in~\cite{HopcroftK73} for bipartite graphs---which we shall use as a blackbox in our algorithms. 

\begin{proposition}[\!\!\cite{DuanP14}]\label{prop:DuanP14}
	There exists an algorithm that given any graph $G=(V,E)$ and any parameter $\eps > 0$, outputs a $(1-\eps)$-approximate maximum weight matching and a $(1+\eps)$-approximate minimum odd-set cover of $G$ 
	in $O(m/\eps)$ time. Moreover, the support of the odd-set cover returned by the algorithm forms a laminar family. 
\end{proposition}

We can now present a proof of~\Cref{cor:time-efficient}. 

\begin{proof}[Proof of Corollary~\ref{cor:time-efficient}]
	The only modification we need for this corollary is that in each iteration of~\Cref{alg:bipartite} (for~\Cref{res:MBM}) or~\Cref{alg:general} (for~\Cref{res:MWM}), instead of finding an exact maximum matching or minimum vertex cover/odd-set cover, 
	we run the algorithm in~\Cref{prop:DuanP14} to find a $(1-\eps)$-approximation to these problems on the sampled graph. 
	
	As for the analysis,~\Cref{lem:Qr-small-bipartite} and~\Cref{lem:Qr-small-general} hold as before as they only relied on the feasibility of the dual solution, not their optimality. On the other hand,~\Cref{lem:Qe-large-bipartite} and
	\Cref{lem:Qe-large-general} relied on the optimality of the dual solutions in each step. However, both these results still hold with a $(1-\eps)$-approximation loss, namely, the conclusion of the lemmas hold as long as we assume
	\[
		\card{{M^{(r)}}} < (1-\eps)^2 \cdot \mu(G) \qquad \text{or} \qquad w({M^{(r)}}) < (1-\eps)^2 \cdot \mu(G,w);
	\]
	this is simply because these conditions, plus the guarantee of~\Cref{prop:DuanP14}, imply the same prior bounds on the value of optimum solution in the sampled graph as needed in~\Cref{lem:Qe-large-bipartite} and~\Cref{lem:Qe-large-general}. 
	The rest of the analysis follows verbatim as before. 
	
	All in all, this ensures that the total runtime of the algorithms is only $\Ot(m/\eps^2 + n/\eps^3)$ time: there are $O(\log{(n)}/\eps)$ passes 
	and each pass takes $\Ot(m/\eps)$ time for reading the stream and sampling the edges plus $\Ot(n/\eps^2)$ time for running the algorithm of~\cite{DuanP14} on the sampled subgraph of size $\Ot(n/\eps)$ edges. This concludes the proof. 
\end{proof}

An alternative way of looking at~\Cref{cor:time-efficient}, from a purely time-complexity point of view, is the following: to obtain a $(1-\eps)$-approximation to maximum (weight) matching on arbitrary graphs, 
we only need to have an oracle for finding a $(1-\Theta(\eps))$-approximation to the same problem on \emph{sparse} graphs with $\Ot(n/\eps)$ edges; running this oracle for $O(\log{(n)}/\eps)$ times on different sampled subgraphs of the input (plus some
low overhead computation for updating the weights and performing the sampling), gives us a $(1-\eps)$-approximation on any (not necessarily sparse) graph as well.

\subsection{Extension to Other Related Models}\label{sec:ext-other-models}
 Our algorithmic approach in this paper is quite flexible and easily extends to many other models.  In particular, given its sample-and-solve nature, the algorithm can be implemented via
a linear sketch (see~\cite{AhnGM12}), which also implies the following two results:

\begin{corollary}[Extension to {Dynamic streams}]
There is a randomized semi-streaming algorithm for $(1-\eps)$-approximation of weighted (general) matching in dynamic streams---with edge insertions and deletions---that for every $p \geq 1$, uses $\Ot(n^{1+1/p}/\eps)$ space and $O(p/\eps)$ passes. 
\end{corollary}

\begin{corollary}[Extension to {Massively Parallel Computation (MPC)}]
There is a randomized MPC algorithm for $(1-\eps)$-approximation of weighted (general) matchings that for every  $p \geq 1$, uses machines of
memory $\Ot(n^{1+1/p}/\eps)$ and $\Ot(n^{1+1/p}/\eps)$  global memory beside the input size, and $O(p/\eps)$ rounds. 
\end{corollary}
As this is not the focus of the paper, we omit the definition and details of the models and instead
refer the interested to~\cite{AhnGM12,McGregor14} and~\cite{KarloffSV10,BeameKS13} for each model, respectively. 
We only note that the prior results in~\cite{AhnG15} also achieved similar corollaries but this is not the case for the approach of~\cite{AssadiJJST22}.

\subsection{Explicit Connections to MWU and Plotkin-Shmoys-Tardos Framework}\label{sec:MWU}

It turns out that our algorithms can be cast in the Plotkin-Shmoys-Tardos (PST) Framework~\cite{PlotkinST91} for solving covering/packing LPs via MWU---despite the fact that the number of iterations is $O(\log{\!(n)}/\eps^2)$ in this framework---by making the following observation (this discussion assumes 
a basic familiarity with this framework; see~\cite{AroraHK12} for a quick introduction). 

Firstly, suppose we use the PST framework for solving the (fractional) vertex/odd-set cover LP, which translates
to maintaining ``MWU weights'' over the edges. The goal, perhaps counter intuitively, is to \emph{fail} in finding a small vertex/odd-set cover, which implies
we have found a ``witness'' to the existence of a large matching in the graph. The \emph{oracle} used for the PST framework
can be implemented by our sampling approach. The key observation is that this oracle is \emph{extremely efficient} compared to a typical oracle, in that, it achieves a very accurate solution with a very small \emph{width}. This leads
to something interesting: the weights of the variables, under the ``cautious'' update rules of MWU, grows so slowly that the \emph{same} oracle solution remains \emph{approximately} valid for the next $O(1/\eps)$ iterations! 

The implication of the above for semi-streaming algorithms is that, effectively, one only needs to 
rerun the oracle, using another pass over the stream, for every $O(1/\eps)$ iterations of the framework. This allows running the $O(\log{\!(n)}/\eps^2)$ iterations of this framework in $O(\log{\!(n)}/\eps)$ passes. 

We shall however caution the reader that while the above intuition is morally true, implementing the algorithm and following the standard analysis this way is quite ``messy'' and does not seem to yield to a necessarily simple algorithm nor analysis. 
Thus, we find the direct proof presented in the paper much more illuminating and opted to provide that instead\footnote{We should also add that this connection was  only made \emph{in hindsight} after having the new algorithm and analysis.}. 


\section{Concluding Remarks and Open Questions}\label{sec:concluding}

In this paper, we presented a rather complete simplification of the prior $O(\log{\!(n)}/\eps)$-pass algorithms of~\cite{AhnG15} in our~\Cref{res:MBM} and~\Cref{res:MWM}. 

The key open question at this point---which was also the key motivation behind this work itself---is to obtain the best of both worlds among the two types of pass-complexity of semi-streaming algorithms obtained for bipartite (cardinality) matching: the $O(\log{\!(n)}/\eps)$ passes of~\cite{AhnG15,AssadiJJST22} and~\Cref{res:MBM}, and the $O(1/\eps^2)$ passes of~\cite{AhnG11,AssadiLT21}. 
\begin{quote}
	\textbf{Open question 1:} \emph{Can we design a semi-streaming algorithm for maximum bipartite cardinality matching with $O(1/\eps)$ passes?}
\end{quote}

We would like to make a (rather bold) conjecture that the \emph{``right''} pass-complexity of this problem might be even (much) lower than $O(1/\eps)$ passes. But, at this point, we seem to be 
far from achieving such results or ruling out their possibilities\footnote{\cite{AssadiJJST22} provides a semi-streaming $n^{3/4+o(1)}$-pass algorithm for solving MBM exactly, which can be seen as $(1-\eps)$-approximation for $\eps=1/n$. Thus, 
at least for  very small values of $\eps$, we already know $\ll 1/\eps$-pass algorithms.}.  

A slightly less exciting question than the above is to \emph{significantly} reduce the pass-complexity of the ``constant-pass'' algorithms for maximum matching in general graphs
in~\cite{FischerMU22,HuangS23} to match the results for MBM in~\cite{AssadiLT21} (see~\Cref{footnote:eps}).  
In particular,  
\begin{quote}
	\textbf{Open question 2:} \emph{Can we obtain a semi-streaming algorithm for maximum matching in general graphs (weighted or unweighted) with $O(1/\eps^2)$ passes?}
\end{quote}

We hope that by simplifying the state-of-the-art, our results in this paper can pave the way for addressing these questions. Note that as stated earlier, we can indeed provide a positive answer to both questions for the (strictly) \emph{more relaxed}
case when the space of the algorithms is  $n^{1+\delta}$ for  constant $\delta > 0$ (this was also already known by the prior work of \cite{AhnG15}). 



\section*{Acknowledgement} 

Many thanks to Soheil Behnezhad, Shang-En Huang, Peter Kiss, Rasmus Kyng, and Thatchaphol Saranurak for helpful discussions and to
the organizers of the DIMACS Workshop on ``Modern Techniques in Graph Algorithms'' (June 2023)---Prantar Ghosh, Zihan Tan, and Nicole Wein---where these conversations happened. 

I am also grateful to Aaron Bernstein, Aditi Dudeja, Arun Jambulapati, Michael Kapralov, Sanjeev Khanna, Kent Quanrud, and Janani Sundaresan for various helpful discussions about this problem over the years. 

Additionally, I would like to thank Soheil Behnezhad for pointing out the similarity of~\Cref{res:MBM} with the $O(n^{1.5}/\eps)$-space $O(1/\eps)$-pass streaming algorithm of~\cite{BehnezhadDETY17} for MBM (as discussed in~\Cref{sec:ext-few-pass}),
and to Vikrant Ashvinkumar and Lap Chi Lau for many helpful comments on the presentation of this paper. 
 
 Finally, I am quite grateful to the anonymous reviewers of TheoretiCS for numerous detailed and insightful comments, and for suggesting simpler proof of~\Cref{lem:Qe-large-bipartite} (and by extension~\Cref{lem:Qe-large-general}) that are 
 presented in the current version of the paper. 

\printbibliography

\end{document}